\documentclass[a4paper,onecolumn,11pt,accepted=2025-05-25]{quantumarticle}
\pdfoutput=1
\usepackage[utf8]{inputenc}
\usepackage[english]{babel}
\usepackage[T1]{fontenc}
\usepackage{amsmath}
\usepackage{hyperref}
\usepackage{doi}
\usepackage{tikz}
\usepackage{lipsum}

\usepackage{commath}
\usepackage{braket}
\usepackage{multicol}
\usepackage{enumerate}
\usepackage{mathrsfs}
\usepackage{amsthm,amssymb,xspace}
\usepackage{cite}
\usepackage{txfonts}
\usepackage{graphics}
\usepackage{float}
\usepackage{epstopdf}
\usepackage{epsfig}
\usepackage{caption}
\usepackage[labelformat=simple]{subcaption}
\usepackage[pass,showframe]{geometry} 
\usepackage{pgfplots}
\usepgfplotslibrary{polar}
\usepgflibrary{shapes.geometric}
\usetikzlibrary{calc}
\usepackage[normalem]{ulem}
\usepackage{booktabs}
\usepackage{commath}
\usepackage{tikz}
\usepackage{tikz-3dplot}
\usepackage[qm]{qcircuit}
\usepackage{makecell}
\usepackage{multirow}

\usepackage{url}
\usepackage[ruled,vlined,linesnumbered]{algorithm2e}
\usepackage{algorithmic}

\newcommand{\mcW}{\mathcal{W}}
\newcommand{\mcM}{\mathcal{M}}
\newcommand{\rt}{\mathsf{rt}}
\newcommand{\rtg}{\mathsf{rt}(G)}
\newcommand{\doh}{\mathsf{doh}}
\newcommand{\dohg}{\mathsf{doh}(G)}

\usepackage{environ}

\newcommand{\repeattheorem}[1]{%
  \begingroup
  \renewcommand{\thetheorem}{\ref{#1}}%
  \expandafter\expandafter\expandafter\theorem
  \csname reptheorem@#1\endcsname
  \endtheorem
  \endgroup
}

\NewEnviron{reptheorem}[1]{%
  \global\expandafter\xdef\csname reptheorem@#1\endcsname{%
    \unexpanded\expandafter{\BODY}%
  }%
  \expandafter\theorem\BODY\unskip\label{#1}\endtheorem
}

\newtheorem{theorem}{Theorem}

\newtheorem{lemma}[theorem]{Lemma}

\begin{document}

\title{Full Characterization of the Depth Overhead for Quantum Circuit Compilation with Arbitrary Qubit Connectivity Constraint}

\author{Pei Yuan}\thanks{peiyuan@tencent.com}
\author{Shengyu Zhang}\thanks{shengyzhang@tencent.com}
\affiliation{Tencent Quantum Laboratory, Tencent, Shenzhen, Guangdong 518057, China}

\maketitle

\begin{abstract}
 In some physical implementations of quantum computers, 2-qubit operations can be applied only on certain pairs of qubits. Compilation of a quantum circuit into one compliant to such qubit connectivity constraint results in an increase of circuit depth. Various compilation algorithms were studied, yet what this depth overhead is remains elusive. In this paper, we fully characterize the depth overhead by the \textit{routing number} of the underlying constraint graph, a graph-theoretic measure which has been studied for 3 decades. We also give reduction algorithms between different graphs, which allow compilation for one graph to be transferred to one for another. These results, when combined with existing routing algorithms, give asymptotically optimal compilation for all commonly seen connectivity graphs in quantum computing.
\end{abstract}

\section{Introduction}\label{sec:intro}

 Quantum computation has demonstrated a substantial advantage over its classical counterpart in solving significant problems such as integer factorization \cite{shor1999polynomial}, search \cite{grover1996fast} and a wide variety of others, by structurally designed algorithms {\cite{qzoo}} as well as variational algorithms {\cite{Cerezo2021-zn}}. 
Quantum hardware technologies have seen considerable advancements in recent years  \cite{IBMQ,ye2019propagation,arute2019quantum,gong2021quantum,ciorga2000addition,elzerman2003few,petta2004manipulation,schroer2007electrostatically,zajac2016scalable,childs2019circuit},
enabling the execution of quantum algorithms. A crucial aspect of this execution involves the compilation of quantum algorithms into quantum circuits, typically composed of $1$- and $2$-qubit gates.

Despite the theoretically proven advantages, the practical implementation of quantum algorithms and quantum circuits faces numerous challenges, one of which is the qubit connectivity constraint. On certain hardware platforms such as superconducting {\cite{IBMQ,arute2019quantum,gong2021quantum}}, quantum dots \cite{ciorga2000addition,elzerman2003few,petta2004manipulation,schroer2007electrostatically,zajac2016scalable} and cold atoms \cite{bloch2008quantum,buluta2011natural,bernien2017probing}, $2$-qubit gates can only act on certain pairs of qubits, while operations on two distant qubits are usually achieved by a sequence of SWAP operations {\cite{booth2018comparing,maslov2008quantum,shafaei2013optimization,shafaei2014qubit,kole2017new,lin2014paqcs,shrivastwa2015fast,li2019tackling}}. This qubit connectivity constraint can be naturally modeled by a connected \textit{constraint graph} $G=(V,E)$, where the vertex set $V$ represents the qubits, and the edge set $E$ specifies the pairs of qubits on which 2-qubit gates can act. A circuit respecting the $G$ constraint is termed \textit{$G$-compliant} in this paper.

Various constraint graphs exist on real hardware, including path \cite{IBMQ, kelly2015state}, bilinear chains \cite{IBMQ,ye2019propagation}, 2D-grids~\cite{arute2019quantum,gong2021quantum}, brick walls~\cite{IBMQ}, cycle-grids \cite{rigetti} and trees~\cite{IBMQ}. Future connectivity patterns may even broaden this variety. 
The wide diversity of the constraint graphs calls for systematic studies on the compilation overhead due to the connectivity constraint. In this paper, we focus on circuit depth, which typically corresponds to the circuit's execution time. In the NISQ era \cite{Preskill2018quantumcomputingin}, the execution time is particularly relevant as most variational quantum algorithms require executing the same circuit thousands of times and use the sample average to estimate the expectation of an observable on the circuit's final state. 

Several questions arise about the depth overhead. On a given constraint graph $G$, what is the smallest depth overhead, as a measure of the graph, that a compiler can possibly achieve? How do we actually compile a circuit achieving this minimal depth increase? When designing a quantum chip, how should we lay out the qubits to ensure their connectivity facilitates a small circuit depth overhead? In this paper, we address these questions by fully characterizing the depth overhead for \textit{any} given graph, with an explicit compilation algorithm given and the optimality shown. 

A widely adopted generic method compiling a quantum circuit under qubit connectivity constraint is to insert SWAP gates into the original circuit to bring (the states in) two distant qubits together, apply the two-qubit operations, and then move them back \cite{booth2018comparing,maslov2008quantum,shafaei2013optimization,shafaei2014qubit,kole2017new,lin2014paqcs,shrivastwa2015fast,li2019tackling}. This method is also the focus of this paper.

Before diving into details, let us specify the measure for depth overhead. {Take a common universal set of 1-qubit and 2-qubit gates} \footnote{{The 1-qubit gates are not restricted by the qubit connectivity. Common choices for the 2-qubit gates include CNOT, CZ, iSWAP, etc. Note that a SWAP gate can be easily realized by three CNOT gates.}}. Given an $n$-qubit circuit $C$ with depth $d(C)$ assuming all-to-all qubit connectivity, we need to compile it to a circuit $C'$ with depth $d(C')$ that respects the constraint graph $G$. We use the ratio of $d(C')/d(C)$ as the overhead measure for circuit instance $C$, aiming to compile any $C$ with a small overhead. Formally, we study the following measure $\doh$ (standing for \textit{depth overhead})
\begin{align} \label{eq:do}
    \doh (G) := \max_C \min_{C'\sim C: \atop G \text{-compliant}} \frac{d(C')}{d(C)}.
\end{align}
Here the minimum is over all $G$-compliant circuits $C'$ equivalent to $C$ and obtained from $C$ by inserting SWAP gates, and the maximum is over all $n$-qubit circuits $C$. Namely, we hope to find a good compiling algorithm $C\to C'$ such that the depth increase ratio $d(C')/d(C)$ is controlled for any possible input circuit $C$ \footnote{Here we use ratio $d(C')/d(C)$ rather than difference $d(C')-d(C)$ because $d(C')$ increases linearly with $d(C)$ for a generic circuit $C$, while the difference $d(C')-d(C)$ can be arbitrarily large. (Though it is also possible to consider multiple layers of $C$ together in compression, this semantic compression does not have much advantage for a \textit{generic} depth-$d$ circuit $C$, as each layer can be arbitrary and different layers do not admit a good compression. Even for the rare cases admitting significant semantic compression, the computational complexity of finding such a good compression is high.)}.

While there are a few compilation algorithms working for a few specific graph constraints \cite{booth2018comparing,maslov2008quantum,shafaei2013optimization,shafaei2014qubit,kole2017new,lin2014paqcs,shrivastwa2015fast,li2019tackling}, the studies fall short in two aspects. Firstly, the proposed routing algorithms for the specific graphs were not always optimal. For instance, Ref. \cite{Harrigan2021-qo} proposed a QAOA circuit that is hardware-compliant under a grid constraint. The method was to first find a long path in the grid and then to use the known SWAP routing algorithm on the path. For a grid of size $\sqrt{n}\times \sqrt{n}$, this results in a $O(n)$ depth overhead. Later we will show a superior routing algorithm with $O(\sqrt{n})$ depth overhead, and show that it is optimal. Secondly, the compilation algorithms so far are \textit{ad hoc} for different specific graphs, and no systematic studies on general graphs have been conducted. This paper is the first to provide a unified, provably optimal algorithm for \textit{all} graphs. 

\paragraph{Main results.} In this work, we present a unified algorithm for qubit routing for any given constraint graph $G$, with the depth overhead fully characterized by the routing number $\rt(G)$, a well-studied graph theoretic measure with a long history. We demonstrate that for all connected graphs $G$,
\begin{align}\label{eq:doh-rt}
    \doh(G) = \Theta(\rt(G)).
\end{align}

Specifically, on any graph $G$, we provide an algorithm to compile an arbitrarily given circuit $C$ (with no connectivity constraint) into another circuit $C'$ with the depth increase ratio bounded by $O(\rt(G))$ from above. Furthermore, we show that this is the best possible outcome---one cannot compile a generic $C$ with a depth increase factor asymptotically better than $O(\rt(G))$.

As a graph theoretic measure that has been studied for about three decades, the routing number $\rtg$ has been pinned down for many specific graphs $G$ such as paths, grids, trees, complete bipartite graphs, hypercubes, etc. \cite{alon1993routing,zhang1999optimal,banerjee2017new,li2010routing,Nenadov2023-re,banerjee2019sorting}. By combining our algorithm with these known routing methods, we immediately obtain compilation algorithms for quantum circuits under these graph constraints. 

Additionally, we introduce a reduction algorithm between different graphs, enabling us to derive efficient routing algorithms for new constraint graphs from existing algorithms on some basic graphs. To demonstrate this, we construct compilers for IBM's brick-wall graphs \cite{IBMQ} and Rigetti's cycle-grid graphs \cite{rigetti} by reducing SWAP networks on those graphs to the ones on the 2D-grid.

\paragraph{Previous work.} The result in Eq. \eqref{eq:doh-rt} bears resemblance to the canonical work on lower bounding quantum circuit size complexity by the geodesic distance on the unitaries manifold \cite{Nielsen05, Nielsen06}. However, our work diverges in two significant ways: (1) we provide matching upper and lower bounds, and (2) our approach is more operational as it presents an explicit algorithm to compile a given arbitrary circuit $C$ to a $G$-compliant circuit $C'$ with depth $d(C') = O(d(C)\cdot \rt(G))$. 

Ref. \cite{allcock2022does} explores the qubit connectivity constraint for general and special classes of unitaries, such as those for quantum state preparation (QSP), by exploiting techniques in earlier work \cite{10044235} and carefully arranging qubits such that most two-qubit gates act on nearby qubits. The paper shows that, somewhat surprisingly, the qubit connectivity constraint does not significantly increase circuit complexity for these classes of unitaries either in the worst or a generic case. For example, it gives a parametrized QSP circuit of depth $O(2^n/n)$ to generate a given arbitrary $n$-qubit quantum state, while even circuits \textit{without} the qubit connectivity constraint require the depth of the same order of depth \cite{PhysRevA.83.032302}. Though this might suggest that the qubit connectivity constraint does not increase circuit complexity, it is important to note that the worst-case or generic unitary matrices have exponentially high complexity. Thus Ref. \cite{allcock2022does} merely demonstrates that the connectivity constraint does not exacerbate these already complex cases. However, our primary concern in practice lies with efficient (e.g. polynomial depth) circuits---After all, these are the ones to be used in the future. Results in Ref. \cite{allcock2022does} do not provide insight into how the connectivity constraint affects these \textit{efficient} circuits, and particularly do not rule out the possibility that shallow circuits significantly suffer from the constraint. This work shows that, fortunately, this is not the case: If a circuit $C$ with all-to-all connections has depth $d$, then it can always be compiled to a $G$-compliant circuit $C'$ with depth at most $O(d\cdot \rt(G)) = O(dn)$. In particular, if $d(C)$ is a polynomial function of $n$, so is the depth of $C'$. 

Ref.~\cite{childs2019circuit} considered routing for partial permutations, in which one only needs to map $k<n$ vertices $x_i$ to $k$ other vertices $y_i$, and the rest $n-k$ vertices can be mapped arbitrarily. When the circuit depth is concerned with, the paper gave a reduction (``Greedy Depth Mapper'') from a partial routing protocol to circuit compilation. Their algorithm is essentially the same as ours in Lemma \ref{lem:ub-by-matching}. However, this compilation can be very loose as explained after Lemma \ref{lem:ub-by-matching}. They also gave a number of partial routing protocols, which may complement our result: They gave efficient routing methods, which may be utilized by our reduction to obtain efficient compilation methods (though one should also be careful about the difference between partial and full permutation).

\paragraph{Organization.} The rest of this paper is organized as follows. In Section \ref{sec:pre}, we introduce notations and review some previous results.
In Section \ref{sec:characterization}, we show the full characterization of the depth overhead. Then we demonstrate a reduction of routing numbers between different graphs and construct routing algorithms for practice qubit connectivity constraints in Section \ref{sec:reduction}. Finally, we discuss our results in Section \ref{sec:discussion}.

\section{Preliminaries}
\label{sec:pre}
\paragraph{Notation.}  Let $[n]:=\{1,2,\ldots,n\}$. Let $|S|$ denote the size of set $S$.
In this paper we study general undirected graphs $G=(V,E)$, where $V$ is the set of vertices and $E$ is the set of edges. 
We denote the number of vertices by $n$, and sometimes identify the vertex set $V$ with the set $[n]$. For a vertex $u\in V$ and a subset $S\subseteq V$, the neighbor of $u$ inside $S$ is $N_S(u) = \{v\in S: (u,v)\in E\}$. We drop the subscript $S$ and just write $N(u)$ if $S = V$. For a subset $S\subseteq V$, the \textit{induced subgraph} of $G$ on $S$ is 
\begin{equation}
    G|_S = (S,E') \text{ 
with } E'=\{(u,v)\in E: u\in S, v\in S\}.
\end{equation}
The \textit{diameter} of a graph $G$ is the largest distance of two vertices, denoted by  
\begin{equation}
  diam(G):=\max_{u,v\in V}d(u,v),  
\end{equation} 
where $d(u,v)$ is the distance of vertices $u$ and $v$ on graph $G$ (i.e. the number of edges on the shortest path between $u$ and $v$).
A \textit{matching} in an undirected graph $G=(V,E)$ is a vertex-disjoint subset of edges $M\subseteq E$.

\paragraph{Quantum circuit compilation and depth overhead.} The qubit connectivity constraint can be modeled by a connected graph $G = (V,E)$.  
A two-qubit gate can be applied to a pair of qubits $i,j\in V$ if and only if $(i,j)\in E$. We refer to $G$ as the \textit{constraint graph} and a circuit satisfying the above constraint as a \textit{$G$-compliant circuit}. When $G$ is the complete graph $K_n$, we say that the circuit has all-to-all connectivity.

The quantum circuit compilation problem studied in this paper is as follows: Given an $n$-qubit input quantum circuit $C$ consisting of 1- and 2-qubit gates with all-to-all connectivity, and a constraint graph $G$ with $n$ vertices (identified with the $n$ qubits), construct a $G$-compliant circuit $C'$ equivalent to $C$ by adding SWAP gates. Here two circuits are equivalent if they implement the same unitary operation. 
To measure the quality of the hardware-compliant circuit, we introduce the following concept of \textit{depth overhead}: 
    \begin{equation}
      \doh(G,C) := \min_{C'} d(C')/d(C),  
    \end{equation}
    where $d(C)$ and $d(C')$ denote the depth of circuits $C$ and $C'$, respectively, and the minimum is over all $C'$ that satisfy the above compilation requirement. The \textit{depth overhead} of a constraint graph $G$ is then defined as
    \begin{equation}
      \doh(G):=\max_{C}\doh(C,G). 
    \end{equation}
where the maximum is taken over all $n$-qubit circuits $C$ with all-to-all qubit connectivity. 

\paragraph{Permutations.}
The set of all permutations on set $[n]$ is denoted by $S_n$. A permutation $\pi\in S_n$ is a \textit{transposition} if $\pi=(a_1,a_2)(a_3,a_4)\cdots (a_{2k-1},a_{2k})$ with distinct $a_1,a_2,\ldots, a_{2k}$, where $(a_{2i-1},a_{2i})$ means to exchange $a_{2i-1}$ and $a_{2i}$. A permutation $\pi$ is a transposition if and only if it satisfies $\pi^2=id$, the identity permutation. It is well known that any permutation can be written as the composition of two transpositions.

\begin{lemma}[\cite{pinter2010book}]\label{lem:permutation_decom}
    Any $\pi\in S_n$ can be decomposed as $\pi=\sigma_1\circ\sigma_2$, where $\sigma_1,\sigma_2\in S_n$ are transpositions.
\end{lemma}

\paragraph{Routing number $\rt(G)$.} The routing number is defined by the following game {\cite{alon1993routing}}. Given a connected graph $G=(V, E)$ with $n$ vertices, we place one pebble at each vertex $v\in V$, and move the pebbles in rounds. In each round $i$, we are allowed to select a matching $M_i\subseteq E$ and swap the two pebbles at $u$ and $v$ for all edges $(u,v)\in M_i$. For any permutation $\pi\in S_n$, $\rt(G,\pi)$ is the \textit{minimum} number of rounds in which one can move all pebbles from their initial positions $v$ to the destinations $\pi(v)$.  For any graph $G$, the routing number is defined as 
\begin{equation}\label{eq:rt}
  \rt(G)=\max_{\pi\in S_n}\rt(G,\pi) .
\end{equation}

\section{Full characterization of depth overhead}
\label{sec:characterization}

In this section, we present a protocol for constructing a $G$-compliant circuit with the depth overhead of at most $O(\rt(G))$ for any given circuit $C$ in Section \ref{sec:upper}. 
Then we demonstrate that for any $G$, there exists a circuit $C$ with depth overhead at least $\Omega(\rt(G))$ in Section \ref{sec:lower}, thereby offering a complete characterization of the depth overhead of $G$.

Computing the depth overhead for a given $C$ and $G$ turns out to be NP-hard; see Appendix \ref{app:hardness} for the proof. 
However, the asymptotic behavior of $\doh(G)$ can be fully characterized by a graph measure called the \textit{routing number}, denoted by $\rt(G)$, of a graph $G$, as defined in Eq. \eqref{eq:rt}. 

\subsection{Hardware-compliant circuit construction and depth overhead upper bound} 
\label{sec:upper}

In this section, we first present a compiling algorithm by the maximum matching. Second, we give a graph partition algorithm. Third, based on the graph partition algorithm, we present the general compiling algorithm and the depth overhead upper bound.

Before diving into detailed constructions, let us first compare the measures $\dohg$ and $\rtg$, and explain why the upper bound $\doh(G) = O(\rt(G))$ does not immediately follow from their definitions, despite the apparent similarities. For easier comparison, let us formulate $\doh(G)$ in a language of games similar to that for the routing number $\rt(G)$ (Eq. \eqref{eq:rt}): we compile circuit $C$ layer by layer, and for each layer, suppose the two-qubit gates are on pairs $(u_1,v_1)$, $\ldots$, $(u_k, v_k)$, then we need to use SWAP gates to move $u_i$ and $v_i$ next to each other, apply the gate, and move them back. 
This formulation highlights immediate similarities between $\doh(G)$ and $\rt(G)$: (1) Both can be viewed as games in rounds, with each round consisting of SWAP operations. (2) Both measures represent the minimum number of rounds. 

However, also note that there are some \textit{key differences} between the two measures: (1) In circuit compilation, $u_i$ and $v_i$ need to be moved to \textit{next} to each other, while in graph routing, $i$ needs to be moved to $\pi(i)$. Note that given a routing algorithm to move each $i$ to $\pi(i)$ for any permutation $\pi$, it is still hard to move $u_i$ and $v_i$ next to each other because we cannot simply find a neighbor $v_i'$ of $v_i$ and let $\pi(u_i) = v_i'$; for example, if all $v_i$'s have degree $1$ and are all connected to a common ``port'' node $v$ to reach the rest of the graph, then all $v_i'$ equal to $v$, making the map $\pi (u_i) = v$ not a permutation. We will need to handle this type of bottleneck issue in designing the protocol for $\dohg$. (2) In  circuit compilation, it suffices that $u_i$ and $v_i$ are {next} to each other at \textit{some} time step $t_i\le \doh(G)$ (different pairs $(u_i,v_i)$ may have different $t_i$), while in graph routing, all $i$ need to be moved to $\pi(i)$ exactly at time $\rt(G,\pi)$. This gives us some freedom to design the $\dohg$ protocol.

One basic property that will be used later is a linear upper bound of $\rt(G)$ {\cite{alon1993routing}}:
\begin{equation}\label{eq:rt-n}
    \rt(G) \le {3n}.
\end{equation}

A routing protocol induces a SWAP circuit in a natural way: For any permutation $\pi$ on vertices in graph $G$ and any routing protocol in the definition of $\rt(G,\pi)$, if two pebbles at two vertices $i$ and $j$ are swapped, then we apply a swap operation on qubits $i$ and $j$ in the SWAP circuit. Then the following unitary transformation $U_\pi$
\begin{equation}\label{eq:unitary_permutation}
  |x_1\cdots x_n\rangle\xrightarrow{U_\pi} |x_{\pi(1)}\cdots x_{\pi(n)}\rangle, \quad \forall x_i\in\{0,1\}, i\in[n],
\end{equation}
can be realized by a $\rt(G,\pi)$-depth circuit consisting of swap operations. 

Before we give the general compiling algorithm, we first give a lemma which can compile circuits for graphs with a large matching.
\begin{lemma}\label{lem:ub-by-matching}
For any connected graph $G$ with the maximum matching size $\nu$, we can construct $G$-compliant circuits with the depth overhead at most $O(\rt(G)\cdot n/\nu)$. That is, $\doh(G) = O(\rt(G)\cdot n/\nu)$.
\end{lemma}
\begin{proof}
    Fix a given $n$-qubit circuit $C$. For each layer $C_i$ with at least one two-qubit gate, suppose that the 2-qubit gates are $C_{i1},\ldots, C_{ik}$ on pairs $\{(u_1,v_1 ),\ldots ,(u_k,v_k )\}$ of qubits, respectively, where $1\le k\le n/2$ and the $2k$ vertices $u_1,v_1,~\ldots~,u_k,v_k $ are all distinct. Let $\{(x_1,y_1),~\ldots~,(x_{\nu},y_{\nu } )\}$ be a maximum matching of $G$.
    We can compile this layer of $C_i$ to a $G$-compliant circuit in depth $O(\rt(G))$ as follows:
   \begin{enumerate}
       \item Apply all single-qubit gates in $C_i$. 
       
       \item \label{step:to-matching} Pick any permutation $\pi$ that permutes $u_j$ to $x_j$ and $v_j$ to $y_j$, for all $j\in [\nu]$. Run the circuit $U_\pi$ (in Eq.~\eqref{eq:unitary_permutation}) of depth at most $\rt(G)$.

       \item Apply the 2-qubit gates $C_{ij}$ on $(x_j,y_{j})$, for all $j\in[\nu]$;

       \item Run $U_\pi^\dagger$, the reverse process of Step \ref{step:to-matching}. 
   \end{enumerate}
   If $k\le \nu$, then this implements $C_i$ already in depth $2\cdot \rt(G)+2$.  If $k>\nu$, then this implements the first $\nu$ two-qubit gates among $k$ ones. Repeat the last three steps in the above procedure until all 2-qubit gates are handled, which needs $\lceil k/\nu \rceil$ iterations. Each iteration needs at most $2\cdot \rt(G)+1$ depth, thus the overall depth overhead is $1+(2\rt(G)+1)\cdot \lceil k/\nu \rceil=O(\rt(G)\cdot n/\nu)$. 
\end{proof}

Note that Lemma \ref{lem:ub-by-matching} can already give compiling algorithms with depth overhead at most $O(\rt(G))$ for some specific graphs $G$, including 1D-chain, 2D-grid, IBM's brick wall or Rigetti's bilinear cycle, binary tree, etc, all of which have a matching of size $\Theta(n)$. But for graphs with a small matching size (an extreme example is the star graph which has the maximum matching size $\nu = 1$), the bound of $O(\rt(G)\cdot n/\nu)$ is very loose. 

Lemma \ref{lem:ub-by-matching} has a clear intuition that the existence of a large matching facilitates moving the pebble around. Actually, it is even tempting to conjecture that this dependence of $O(n/\nu)$ is inevitable since a bottleneck in a graph does make simultaneous pebble moving inefficient due to traffic congestion. However, this bottleneck also affects $\rtg$ protocols and should be inherently characterized in the $\rtg$ measure. What we need to do is to technically relate the difficulty in the two measures and construct a reduction from one to the other. Next, we give details on how to remove the $O(n/\nu)$ factor in Lemma \ref{lem:ub-by-matching}.

To improve it to the optimal bound of $O(\rt(G))$, 
one idea is to partition the constraint graph into vertex-disjoint connected subgraphs, each having $O(1)$ diameter and containing at most $O(\rt(G))$ vertices. We aim to move each pair of qubits $(u_i,v_i)$ on which a two-qubit gate acts to one of these subgraphs (different pairs may go to different subgraphs). 
If this can be achieved, then we can implement two-qubit gates within one subgraph efficiently. Indeed, since each subgraph has diameter $O(1)$, it takes $O(1)$ SWAP gates to implement one gate and since the subgraph has size $\rt(G)$, all the gates inside this subgraph can be done by $O(\rt(G))$ SWAP gates, which takes at most $O(\rt(G))$ rounds. Also note that the routings in different subgraphs can be carried out in parallel. Thus the overall overhead is $O(\rt(G))$.

The challenge is that it is not always possible to achieve such a good partition of the constraint graph. 
We present a good graph partition algorithm in Lemma \ref{lem:graph_partition}, for which we will first show the following bottleneck lemma.

\begin{lemma}[bottleneck] \label{lem:bottleneck} 
    For a connected graph $G=(V,E)$, suppose that there exist vertex sets $V_1,V_2\subseteq V$ such that 
    \begin{enumerate}
        \item $V_1\cap V_2=\emptyset$;

        \item for any $u\in V_1$, $N(u)\subseteq V_2$, where $N(u):=\{w\in V: (w,u)\in E \}$.
    \end{enumerate}
    Then the routing number of $G$ satisfies 
    \begin{equation}
        \rt(G)=\Omega\left(\frac{|V_1|}{|V_2|}\right).
    \end{equation}
\end{lemma}
\begin{proof}
    Let us label the pebbles in $V_1$ as $1,2,\ldots, |V_1|$. Let $s:=\lfloor|V_1|/2\rfloor$. Define a permutation $\pi:=(1,1+s)(2,2+s)\cdots (s,2s)$. Note that if we move the pebble at vertex $i$ to the vertex $i+s$, the pebble must go through $V_2$ since $N(i)\subseteq V_2$ by assumption. Since at most $|V_2|$ pebbles in $V_1$ can go through $V_2$ in one round, moving $s$ pebbles needs at least $s/|V_2|$ rounds. Therefore, 
    \[\rt(G)\ge \rt(G,\pi)\ge s/|V_2|=\Omega(|V_1|/|V_2|).\]
\end{proof}

Next, we present the graph partition algorithm in Lemma \ref{lem:graph_partition}, which outputs two families of vertex sets $\mathcal{W}$ and $\mathcal{W}'$. 
\begin{lemma}\label{lem:graph_partition}
 There exists an algorithm which, on any $n$-vertex connected graph $G=(V,E)$, outputs two families of vertex sets 
 \begin{align}
     \mathcal{W} &=\{W_1,\ \ldots,\ W_s: W_i \subseteq V,\ \forall i\in[s]\},\\
     \mathcal{W'} &=\{W_1',\ \ldots,\ W_t': W_j' \subseteq V,\ \forall j\in[t]\},
 \end{align}
 for some $s,t\in[n]$, satisfying the following properties.
\begin{enumerate}
    \item (disjointness) For any distinct $i,i'\in[s]$ and distinct $j,j'\in [t]$, we have $W_i\cap W_{i'}=\emptyset$ and $W'_j\cap W'_{j'}=\emptyset$; \label{prop: disjointness}
    
    \item (coverage) $|\bigcup_{i\in[s]}W_i|+|\bigcup_{j\in[t]}W'_j|\ge n/2$; \label{prop: coverage}

    \item   (size bound) For any $i\in[s]$ and $j\in[t]$, we have $2\le |W_i|=O(\rt(G))$ and $2\le |W'_j|=O(\rt(G))$; \label{prop: size_ub}

    \item (diameter bound) For all $i\in[s]$ and $j\in [t]$, the induced subgraphs $G_i = G|_{W_i}$ and $G_j' = G|_{W_j'}$ all have diameter at most $2$.\label{prop: diameter} 
\end{enumerate}   
The algorithm runs in time $O(n^3)$.
\end{lemma}
\begin{proof}
The family $\mathcal{W}$ is constructed as follows. First find a maximal matching  
\begin{equation}
\mathcal{M} 
=\{(w_{1},w_{2}),\ (w_3, w_4),\ \ldots,\ (w_{2s-1}, w_{2s})\}
\end{equation}
of $G$. 
Put 
\begin{equation}
   \mathcal{W} = \{W_i: i \in[s]\}, \text{ where } W_i = \{w_{2i-1},w_{2i}\}\subseteq V, \ \forall i\in [s]. 
\end{equation}
Then $\mathcal{W}$ satisfies the properties \ref{prop: disjointness}, \ref{prop: diameter} and \ref{prop: size_ub}. Indeed, each $G|_{W_i}$ is essentially an edge and thus the two nodes are connected, different $G|_{W_i}$'s are disjoint as $\mcM$ is a matching, and $|W_i| = 2 = O(\rt(G))$. 

The family $\mathcal{W}'$ is constructed as follows. Define vertex set $T = V - \cup_{i\in[s]} W_i$, the vertices not in the maximal matching $\mcM$. Note that $T$ is an independent set, i.e. any two vertices $a,b\in T$ are not connected; otherwise, we could have added the edge $(a,b)$ in $\mathcal{M}$ to form a larger matching, contradicting $\mcM$ being maximal. 
Define set
\begin{equation}
S:=\big\{w\in  \bigcup_{i\in[s]}W_i: N_T(w) \neq \emptyset\big\}
\end{equation}
to contain those vertices with a connection to $T$. The family $\mathcal{W}'$ is constructed by Algorithm \ref{alg:set_construction}. 
\begin{algorithm}[!htb]
    \caption{Construction of $\mathcal{W}'$}
    \label{alg:set_construction}
    \textbf{Input}: Connected graph $G=(V,E)$, vertex sets $T,S\subseteq V$.\\
    \textbf{Output}: A family $\mathcal{W}'$ of sets.\\
    \begin{algorithmic}[1]
        \STATE Initialize $N^0_{w} := \emptyset$, $\forall w\in S$.
        \STATE Initialize $A_{0,p} := \emptyset$, $\forall p\in[|T|]$.
        \STATE $T_1:=T$, $S_1:=S$, $k_1:=|S_1|$.
        \FOR{ $p=1$ to $|T|$ }\label{line:outerloop}
            \FOR{$i=1$ to $|S|$}\label{line:innerloop} 
                \STATE $A_{i,p}:=\emptyset$.
                \IF {there are at least 1 and  at most $\lceil |T_p|/k_{p}\rceil$ neighbors of $w_i$ in $T_p - \bigcup_{r=1}^{i-1} A_{r,p}$} \label{line:condition1}
                    \STATE Let $A_{i,p}$ contain all these neighbors.\label{line:set1}
              \ELSIF{there are more than $\lceil |T_p|/k_{p}\rceil$ neighbors of $w_i$ in $T_p - \bigcup_{r=1}^{i-1} A_{r,p}$} \label{line:condition2}
                    \STATE Let $A_{i,p}$ contain arbitrary $\lceil |T_p|/k_{p}\rceil$ many of these neighbors.\label{line:set2}
              \ENDIF
              \STATE Let $N^p_{w_i}:=N^{p-1}_{w_i}\cup A_{i,p}$. \label{line:update N}
            \ENDFOR\label{line:innerloopend}
            \IF {$|\bigcup_{i=1}^{|S|} N^p_{w_i}| {\ge |T|/2} $} \label{line:stop-condition}
                \RETURN $\mathcal{W}':=\{{N^p_{w_i}}\cup \{w_i\}: |{N^p_{w_i}}| \ge 1,\ i \in[|S|]\}$ and end the whole program. \label{line:return}
            \ENDIF
            \STATE Set $T_{p+1}:=T_p\backslash  \bigcup_{r=1}^{|S|} A_{r,p}$. \label{line:update T}
            \STATE Set $S_{p+1}:= \{ w_i: |A_{i,p}|=\lceil |T_p|/k_p\rceil, i\in[|S|]\}$. \label{line:sp}
            \STATE Set $k_{p+1}:=|S_{p+1}|$.\label{line:kp}
        \ENDFOR
    \end{algorithmic}
\end{algorithm}

In the algorithm, $T_p$, $A_{i,p}$, $S_p$, $k_p$ and $N_{w_i}^p$ are defined as follows. The set $T_p$ contains those vertices in $T$ not selected in the first $p-1$ iterations. 
In the $p$-th iteration, the set $A_{i,p}$ denotes the neighbor set within $T_p-\bigcup_{r=1}^{i-1}A_{r,p}$ of vertex $w_i\in S$ with cardinality bounded by $\lceil|T_p|/k_p\rceil$.
The set $S_p$ contains all vertices $w_i\in S$ that have exactly $\lceil |T_{p-1}|/k_{p-1}\rceil$ neighbors being chosen in the $(p-1)$-th iteration.
The number $k_p$ denotes the size of vertex set $S_p$. 
The set $N_{w_i}^p$ contains all neighbors of vertex $w_i$ chosen in the first $p$ iterations.

We will first show that $\mathcal{W}'$ satisfies its corresponding properties in \ref{prop: disjointness} (disjointness) and \ref{prop: diameter} (diameter). 
For the disjointness property, we note that actually all $A_{i,p}$'s, for different $i$ and different $p$, are pairwise disjoint. Indeed, in each outer iteration $p+1$, the new set $T_{p+1}$ removes all sets $A_{r,p}$ in iteration $p$ (line \ref{line:update T}), thus any set $A_{r',p+1}$ selected from $T_{p+1}$ is disjoint from all sets $A_{r,p'}$ from previous iterations $p'\le p$. Now we check the sets $A_{i,p}$ for different $i$ in the same iteration $p$. Note that when we consider neighbors of $w_i$, we ignore those in previous inner iterations by only considering $T_p - \bigcup_{r=1}^{i-1} A_{r,p}$ (line \ref{line:condition1} and \ref{line:condition2}), thus the new $A_{i,p}$ are disjoint from $A_{r,p}$ for all $r<i$. 
Now that all $A_{i,p}$'s are pairwise disjoint, and all $w_i$'s are distinct, the sets $N_{w_i}^p\cup \{w_i\}$ in the definition of $\mcW'$ are pairwise disjoint as well. This shows the property of disjointness. 

As shown in lines \ref{line:set1} and \ref{line:set2}, all vertices in $A_{i,p}$ are connected with vertices $w_i$, which implies that all vertices in $N_{w_i}^p$ (line \ref{line:update N}) are connected with $w_i$. Therefore, the subgraph induced by $N_{w_i}^p \cup \{w_i\}$ has diameter at most $2$. 

With the above, we can next show that the algorithm ends and returns $\mcW'$ in line \ref{line:return}, i.e. the condition in line \ref{line:stop-condition} is satisfied at some iteration $p=\ell$, in at most $|T|/2 \le n/2$ iterations. We will show this by arguing that in each outer iteration $p$, each vertex $v\in T_p$ has at least one neighbor in $S_p$. Once we show this, we know that at least one $A_{i,p}$ is nonempty, for which $N_{w_i}^p$ has size strictly larger than that of $N_{w_i}^{p-1}$ (line \ref{line:update N}): All $A_{i,p}$'s are disjoint as shown above, thus $A_{i,p} \cap N_{w_i}^{p-1} = \emptyset$ and thus $|N_{w_i}^{p}| = |N_{w_i}^{p-1}| + |A_{i,p}| > |N_{w_i}^{p-1}|$. Therefore the set $\bigcup_{i=1}^{|S|} N^p_{w_i}$ strictly increases its size as $p$ grows. Thus the condition in line \ref{line:stop-condition} is met and the algorithm ends after at most $|T|/2$ outer iterations. 

Now we argue by induction that in each outer iteration $p$, any vertex $v\in T_p$ has at least one neighbor in $S_p$, and all neighbors of $v$ are in $S_p$. Namely, we have
\begin{equation}\label{eq:Nv_in_Sp}
\emptyset \neq N(v) \subseteq S_p, 
\end{equation}
for all $ p\in [|T|]$ and all $v\in T_p$. This is true for $p=1$ as $v$ does not have neighbors in $T_1 = T$, thus all its neighbors are in $W = \{w_1,\ w_2,\ \ldots, \ w_{2s}\}$. Furthermore, all its neighbors are actually in $S\subseteq W$, as $S$ exactly contains those $u\in W$ that have neighbors in $T$, i.e. $W-S$ does not have any edge to $T$. Therefore, $v$ has at least one neighbor and all $v$'s neighbors are in $S_1 = S$. For the induction step, let us assume $\emptyset \neq N_{S_p}(v) $ for each $v\in T_p$, and consider iteration $p+1$. For each $v\in T_{p+1}\subseteq T_p$, its neighbors $w_{i_1}, \ldots, w_{i_k}$ are all in $S_p$ and $k\ge 1$ by inductive hypothesis. But this $v$ is selected in line \ref{line:update T} to be in $T_{p+1}$. This happens must because it is not in $A_{r,p}$ for any $r\in |S|$, including $A_{i_j,p}$. As each $w_{i_j}$ has at least one neighbor $v\in T_p$, we know that the condition in line \ref{line:condition2} is satisfied, i.e. $w_{i_j}$ has more than $\lceil |T_p|/k_{p}\rceil$ neighbors in $T_p - \bigcup_{r=0}^{i_j-1} A_{r,p}$, but it then happens in line \ref{line:set2} that $v$ is not chosen into $A_{i_j,p}$. This means that $|A_{i_j,p}|=\lceil |T_p|/k_p\rceil$ and thus $w_{i_j}\in S_{p+1}$ (line \ref{line:sp}). Therefore, $v$'s neighbors $w_{i_j}$ are all in $S_{p+1}$, completing the inductive step.

We can also show that the algorithm never runs into the situation of $S_{p+1} = \emptyset$ in line \ref{line:sp} (and $k_{p+1}$ is always nonzero in the next line, justifying it being denominator in lines \ref{line:condition1}-\ref{line:set2}). Indeed, if $S_{p+1} = \emptyset$, it means that all $w_i\in S_p$ has $|A_{i,p}|\le \lceil |T_p|/k_{p}\rceil-1$ and line \ref{line:set1} is executed. But then $\bigcup_{i=1}^{|S|} N^p_{w_i}$ is the entire $T$, and thus the condition in line \ref{line:stop-condition} is satisfied and algorithm returns $\mcW'$ in line \ref{line:return} before line \ref{line:sp}. 

Next we show that $\mathcal{W}'$ satisfies property \ref{prop: size_ub}, namely $2\le |W'_j|=O(\rt(G))$. Assume that Algorithm \ref{alg:set_construction} stops in the $\ell$-th iteration of the outer loop. Since $N^\ell_{w_i}\cup \{w_i\}$ in $\mathcal{W}'$ satisfies $|N^\ell_{w_i}|\ge 1$ (line \ref{line:condition1} and \ref{line:condition2}), each set in $\mathcal{W}'$ has size at least $2$. We next prove that it has size at most $O(\rt(G))$.
           
Suppose $S_{\ell}:=\{w_{i_1}, \ldots, w_{i_{k_{\ell}}}\}$. Consider the process of obtaining $S_1, \ldots, S_\ell$ in the algorithm. In each outer iteration $p$, the algorithm checks each $w\in S_p$ and selects as many neighbors as possible, but up to $\lceil |T_p|/k_{p}\rceil$, to form $A_{i,p}$. If there are more than $\lceil |T_p|/k_{p}\rceil$ neighbors, then it continues to collect these neighbors in the next outer iteration. $S_p$ contains those $w_i\in S_{p-1}$ with $|A_{i,p-1}| = \lceil|T_{p-1}|/k_{p-1}\rceil$. So for any $w_i\in S_\ell$, the size of $A_{i,p}$, for $p=1, 2, \ldots, \ell-1$, is $\lceil|T_1|/k_1\rceil, \lceil|T_2|/k_2\rceil, \ldots, \lceil|T_{\ell-1}|/k_{\ell-1}\rceil$, respectively.  
And the size of the corresponding set $N^{\ell-1}_{w_{i}}$ has 
\begin{equation}\label{eq:size_set}
    |N^{\ell-1}_{w_{i}}|=\big|\bigcup_{r\in [\ell-1]} A_{i,r}\big|=\sum_{r=1}^{\ell-1}|A_{i,r}|=\sum_{r=1}^{\ell-1}\lceil |T_{r}|/k_{r}\rceil. 
\end{equation}
where the second equality uses the fact that all $A_{i,r}$'s are disjoint. 

{Define $C:=\bigcup_{w\in S_\ell}N_{w}^{\ell-1}$, the union of the corresponding sets of vertices in $S_\ell$.}  
Based on Eq. \eqref{eq:size_set} and the fact that all these $N_{w}^{\ell-1}$'s are disjoint, we have
\begin{equation}\label{eq:size_C}
    |C| = \sum_{w\in S_\ell}|N_{w}^{\ell-1}|=k_\ell\cdot \sum_{r=1}^{\ell-1}\lceil |T_{r}|/k_{r}\rceil.
\end{equation}
Since Algorithm \ref{alg:set_construction} did not stop in the $(\ell-1)$-th step, we have $\big| \bigcup_{w\in S}N_{w}^{\ell-1} \big|<|T|/2$. Then 
\begin{align*}
    |C| = &\big| \bigcup_{w\in S_\ell} N_{w}^{\ell-1} \big| 
    \le \big| \bigcup_{w\in S} N_{w}^{\ell-1} \big|<|T|/2,\\
    |T_\ell| = &\big|T_{\ell-1} - \bigcup_{r=0}^{|S|} A_{r,\ell-1}\big| = \cdots = \big|T - \bigcup_{p=1}^{\ell-1}\bigcup_{r=0}^{|S|} A_{r,p}\big| = \big|T - \bigcup_{r=0}^{|S|}N_{w_{r}}^{\ell-1}\big|\\
    \ge & |T|-|T|/2=|T|/2,
\end{align*}
Therefore $|T_\ell|\ge |T|/2>|C|$. 
Since $T_\ell\subseteq T$, $S_\ell\subseteq S$, it follows that $T_\ell\cap S_\ell=\emptyset$. We have showed that $N(u)\subseteq S_\ell$ for all $u\in T_\ell$ by Eq. \eqref{eq:Nv_in_Sp}, and thus can apply Lemma \ref{lem:bottleneck} to obtain that $\rt(G)= \Omega(\lceil|T_\ell|/|S_\ell|\rceil)=\Omega(\lceil|T_\ell|/k_\ell\rceil)$. 
Combined with Eq. \eqref{eq:size_C}, we have
\begin{align*}
    \rt(G)=& \Omega(\lceil|T_\ell|/k_{\ell}\rceil)  
     \ge \Omega(\lceil|C|/k_{\ell}\rceil) 
     =  \Omega\left(k_\ell\cdot \sum_{r=1}^{\ell-1}\lceil |T_{r}|/k_{r}\rceil/k_\ell\right)=\Omega\left(\sum_{r=1}^{\ell-1}\lceil |T_{r}|/k_{r}\rceil\right).
\end{align*}
Therefore, the routing number $\rt(G)$ satisfies
\begin{align*}
    \rt(G) = & \max\left\{\Omega\left(\sum_{r=1}^{\ell-1}\lceil |T_{r}|/k_{r}\rceil\right),\Omega\left(\lceil|T_\ell|/k_{\ell}\rceil\right)\right\} \\
= &\Omega\left(\sum_{r=1}^{\ell-1}\lceil |T_{r}|/k_{r}\rceil+\lceil|T_\ell|/k_{\ell}\rceil\right) \\
= & \Omega\left(\sum_{r=1}^{\ell}\lceil |T_{r}|/k_{r}\rceil\right).
\end{align*}

For arbitrary set $N_{w_i}^\ell\cup \{w_i\}$ in $\mathcal{W}'$,
\begin{align*}
    & |N_{w_i}^\ell\cup \{w_i\}| \\
    = &|N^{\ell-1}_{w_i}\cup A_{i,\ell}|+|\{w_i\}|=|N^{\ell-1}_{w_i}|+|A_{i,\ell}|+|\{w_i\}|,\\
    \le & \sum_{r=1}^{\ell}\lceil |T_{r}|/k_{r}\rceil+1,& (\text{Eq. }\eqref{eq:size_set})\\
    = & O(\rt(G)). & (\rt(G)= \Omega(\sum_{r=1}^{\ell}\lceil |T_{r}|/k_{r}\rceil))
\end{align*}
Therefore, $\mathcal{W}'$ satisfies property \ref{prop: size_ub}.

Recall that $T=V\backslash \bigcup_{i\in[s]} W_i$ and $\mathcal{W} = \{W_i: i\in[s]\}$. In the definition of $\mathcal{W}':=\{{N^p_{w_i}}\cup \{w_i\}: |{N^p_{w_i}}|\ge 1, i \in[|S|]\}$ (line \ref{line:return}), suppose there are $t$ many $i\in [|S|]$ satisfying $|{N^p_{w_i}}|\ge 1$, and denote the sets ${N^p_{w_i}}\cup \{w_i\}$ as $W_1', \ldots, W_t'$. The condition in line \ref{line:stop-condition} of Algorithm \ref{alg:set_construction} implies that 
\begin{equation}
    |\bigcup_{i\in[s]}W_i|+|\bigcup_{j\in[t]}W'_j|\ge n-|T|+ |T|/2=n-|T|/2\ge n/2,
\end{equation}
showing property \ref{prop: coverage} in the theorem.

Finally, we analyze the complexity. {Finding a maximal matching can be easily done in time $O(n^3)$ by repeatedly adding an edge $(u,v)$ into the matching set $M$ and removing $u$ and $v$ from the vertex set.} In the $i$-th inner loop, lines \ref{line:condition1}-\ref{line:update N} can be realized in time $O(|T|)$. Since there are $|S|$ inner loops, then lines \ref{line:innerloop}-\ref{line:innerloopend} can be realized in time $O(|S|\cdot |T|)$. In the $p$-th outer loop, the updates in lines \ref{line:update T}-\ref{line:kp} can be completed in $O(|T|+|S|)$ time. Since there are $|T|$ outer loops, the total time of Algorithm \ref{alg:set_construction} is \[|T|\cdot (O(|S|\cdot |T|)+O(|S|+|T|))=O(|T|^2|S|)=O(n^3).\] The total time for constructing $\mathcal{W}$ and $\mathcal{W}'$ is $O(n^{3})+O(n^3)=O(n^3)$.

\end{proof}

With this result, we can state and prove the main compilation algorithm next, from which it will also be clear why we need those properties of the two families. 

\begin{theorem}\label{thm:depthoverhead_ub}
    For any connected graph $G$, we have $\doh(G) = O(\rt(G))$. 
\end{theorem}
\begin{proof}
For any graph $G=(V,E)$, we use Lemma \ref{lem:graph_partition} to find two families of sets $\mathcal{W}=\{W_i:~\forall i\in[s]\}$ and $\mathcal{W}'=\{W_j':~\forall j\in[t]\}$ satisfying the properties in the lemma. Since $|\bigcup_{i\in[s]}W_i |+|\bigcup_{j\in[t]}W'_j|\ge n/2$, at least one of $|\bigcup_{i\in[s]}W_i|$ and $|\bigcup_{j\in[t]}W'_j|$ is of size at least $n/4$. Without loss of generality, we assume that $|\bigcup_{i\in[s]}W_i|\ge n/4$. 
{Consider any one layer of a given circuit. Suppose it has $k$ ($k\le \lfloor n/2\rfloor$) 2-qubit gates, which act on pairs $(i_1, j_1), (i_2, j_2), \ldots, (i_k, j_k)$ of qubits. Note that qubits $i_r$ and $j_r$ are generally not adjacent on $G$. We call these the original gates, to distinguish from the compiled ones that are $G$-compliant, which will be referred to as \textit{$G$-gates}.} 

{We now show how to compile one layer of 2-qubit original gates to a $G$-compliant circuit of depth at most $O(\rt(G))$.}
\begin{enumerate}
    \item {Take a permutation $\pi$ that moves the first $\gamma_1 = \lfloor |W_1|/2\rfloor$ pairs $(i_1, j_1), \ldots, (i_{\gamma_1}, j_{\gamma_1})$ to inside $W_1$, and the next $\gamma_2 = \lfloor |W_2|/2\rfloor$ pairs $(i_{\gamma_1+1}, j_{\gamma_1+1}), \ldots, (i_{\gamma_1+\gamma_2}, j_{\gamma_1+\gamma_2})$ to inside $W_2$, and so on, until we move $\gamma_s = \lfloor |W_s|/2\rfloor$ pairs to $W_s$. Implementing this permutation needs at most $\rt(G)$ rounds (see Eq. \eqref{eq:unitary_permutation}). Now the problem of implementing the $\sum_s \gamma_s$ original gates $U_{i_r,j_r}$ is reduced to implementing the corresponding gates $U_{p,q}$ with $p$ and $q$ in some $W_i$.}
    \item Inside $W_i$ for each $i\in [s]$, we can implement all $\gamma_i$ 2-qubit gates $U_{p,q}$ in depth at most $3 \gamma_i$. Indeed, for each gate $U_{p,q}$ on qubits $(p,q)$ in $W_i$, since $diam(G_i)\le 2$, either $(p,q)\in E$ or there is another vertex $r\in W_i$ connecting $p$ and $q$. In the former case, we can apply the gate directly. In the latter case, we can swap $p$ and $r$, apply the gate on $(r,q)$, and swap $p$ and $r$ back. In any case, we can implement one gate by at most three $G$-gates, thus $3\gamma_i$ $G$-gates suffice to implement all $\gamma_i$ 2-qubit gates in $W_i$. Note that each $|W_i| = O(\rt(G))$ by Lemma \ref{lem:graph_partition}, thus the depth is $O(\rt(G))$ by Eq. \eqref{eq:rt-n}.

    \item By repeating the above two steps $\lceil k/(\sum_{i=1}^s \gamma_i)\rceil$ iterations, we can implement all of the original 2-qubit gates. {Since $k\le n/2$, $\gamma_i = \lfloor |W_i|/2 \rfloor$, and $\sum_i |W_i| \ge n/4$, the number of iterations is at most $O(1)$.} 
    
    \item Move all qubits to their original positions, which takes at most $\rt(G)$ rounds.
\end{enumerate}
Putting everything together, the depth overhead for compiling one layer of the original circuit is at most $O(\rt(G))$. Applying this to all layers completes the proof. 
\end{proof}

\paragraph{Remark.} Though the above theorem is only on the depth overhead, from the proof we can see that actually 
our construction from an $\rtg$ protocol to a $\dohg$ protocol can be done very efficiently. Since there are efficient routing protocols for many commonly seen specific graphs $G$, our construction enables us to obtain efficient compilation for quantum circuits. To be more specific, given any qubit connectivity graph $G$, we only need $O(n^3)$ classical pre-processing time for finding $\mathcal W$ and $\mathcal W'$ in Lemma \ref{lem:graph_partition}, which needs to be computed only once and can be used for compiling any circuit. For each quantum circuit, the compilation time is merely $O(n)$, because it is easily verified that all steps in the above proof just need to identify some permutation $\pi$ and run the given routing algorithms in $\rtg$.

\subsection{Depth overhead lower bound} 
\label{sec:lower}
The following theorem gives a lower bound of the depth overhead, which matches the upper bound in Theorem \ref{thm:depthoverhead_ub}.

\begin{theorem}\label{thm:depthoverhead_lb}
    For any connected graph $G$, we have $\doh(G) = \Omega(\rt(G))$.
\end{theorem}
\begin{proof}   
Suppose $\rt(G) = \rt(G,\pi^*)$, namely $\pi^*$ is the hardest permutation in the definition of $\rt(G)$. By Lemma \ref{lem:permutation_decom}, $\pi^*$ can be decomposed into two compositions. Thus there is an unconstrained SWAP circuit $C^*$ of depth 2 realizing $\pi^*$. 
By definition of $\doh(G)$, we have 
\[\doh(G) = \max_C \min_{C'\sim C: \atop G\text{-compliant}} \frac{d(C')}{d(C)} \ge \min_{C'\sim C: \atop G\text{-compliant}} \frac{d(C')}{d(C^*)} = \min_{C'\sim C: \atop G\text{-compliant}} \frac{d(C')}{2}.\]
Take a $C'$ achieving the minimum in the above ratio. Since $C^*$ is a SWAP circuit and $C'$ is obtained from $C$ by inserting SWAP gates, $C'$ is also a SWAP circuit. Therefore by the correspondence in Eq. \eqref{eq:circuit-routing-correspondence}, we obtain a routing algorithm $A^*$ of $d(C')$ rounds which can realize $\pi^*$. Since this is one particular routing algorithm realizing $\pi^*$, 
we have 
\[\rt(G) = \rt(G,\pi^*) \le \text{ the number of rounds in }A^* = d(C') = 2\cdot \doh(G).\]
Therefore $\doh(G) \ge \rt(G)/2 = \Omega(\rt(G))$.
\end{proof}

Combining Theorems \ref{thm:depthoverhead_ub} and \ref{thm:depthoverhead_lb}, we see that the depth overhead of a connectivity graph $G$ is fully characterized by its routing number, i.e., $\doh(G)=\Theta(\rt(G))$. 
This result shows that our algorithm in 
Theorem \ref{thm:depthoverhead_ub} is asymptotically optimal. 
The characterization also gives quantitative guidance for the qubit layout and connectivity design of quantum processors when the depth overhead is considered.

\section{Routing number for many common graphs and reduction between graphs}
\label{sec:reduction}

In this section, we demonstrate a reduction of routing numbers between different graphs and construct routing algorithms for cycle-grids and brick walls.

Many routing algorithms for different specific graphs have been widely investigated, see Table \ref{tab:routing_number} for a summary. 
{These algorithms, combined with our general algorithm in Theorem \ref{thm:depthoverhead_ub}, give optimal routing algorithms for many existing connectivity graphs}
such as paths, bilinear chains~\cite{IBMQ,ye2019propagation}, 2-dimensional grids~\cite{arute2019quantum,gong2021quantum}, and trees~\cite{IBMQ}. {This improves some of previous SWAP algorithms. For example, } Ref. \cite{Harrigan2021-qo} proposed a circuit ansatz for QAOA, but to make it hardware-compliant for grid constraints, their algorithm has a depth overhead of $O(n)$ on a $\sqrt{n}\times \sqrt{n}$ 2D-grid. Using Theorem \ref{thm:depthoverhead_ub} and the result for 2D-grids in \cite{alon1993routing}, we easily achieve a depth overhead of $O(\sqrt{n})$, quadratically improving the previous one and being the best possible.

While the routing number for the above graphs has been well studied, it has not been studied for graphs that are less commonly seen in graph theory but typical in quantum computing, such as IBM's brick walls and Rigetti's cycle-grids. 
Lemma \ref{lem:ub-by-matching} can be used to solve some graphs, and here we provide another method based on reduction, which can give good routing algorithms for more graphs.

\begin{table}[]
    \centering
    \caption{A summary for the routing numbers of different graphs. 
    }
    \label{tab:routing_number}
    \resizebox{\textwidth}{!}{\begin{tabular}{c c c c}
    \hline
    \hline
     {\bf Graph $G$}    & ~\bf Number of vertices~ & ~\bf Upper bound of $\rt(G)$~ & ~\bf Reference  \\
     \hline
      \multirow{2}*{Tree $T_n$}   & \multirow{2}*{$n$} & $3n$ & \cite{alon1993routing}\\
                &&$\lfloor 3n/2\rfloor+O(\log(n))$ &\cite{roberts1995routing}\\
    Complete $d$-ary tree $T_n^d$ & $n$ &$n+o(n)$ &  \cite{roberts1995routing}\\
    Complete bipartite graph $K_{s,t}$ ($s>t$) & $s +t$ &$\lfloor 3s/2t\rfloor+7$ &\cite{li2010routing}\\
    $2$-dimensional grid $Grid_{n_1,n_2}$ ($n_1\le n_2$) & $n_1n_2$ & $2n_1+n_2$ & \cite{alon1993routing}\\
    Path $P_n$ & $n$ & $n$ & \cite{alon1993routing}\\
    Hypercube $Q^n$ & $2^n$ & $2n-2$  & \cite{alon1993routing}\\
    \hline
    \hline
     \end{tabular}}
\end{table}

We will in particular need one result for the Path graph (i.e. 1D-chain). 
\begin{lemma}[\cite{alon1993routing}]\label{lem:path_rout}
    Let $P_n$ denote a path with $n$ vertices, then $\rt(P_n) =n $.
\end{lemma}

Now we prove Theorem \ref{thm:reduction_routing}, which gives a reduction between routing numbers of two graphs. 
\begin{theorem}\label{thm:reduction_routing}
    Let $G=(V,E)$ and $G'=(V,E')$ be two connected graphs with the same vertex set $V$. 
    Suppose that $E'- E :=\{e: e \in E' \text{~and~} e\notin E\}$ can be partitioned into $ \bigcup_{i=1}^c E_i'$ such that the following two conditions hold.
    \begin{enumerate}
        \item $E_i' \cap E_j' = \emptyset$ for arbitrary distinct $i,j\in [c]$.
    
        \item $\rt(G,\pi_i) \le c'$ for each $i\in [c]$, where $\pi_i = \circ_{(u,v)\in E_i'} (u,v)$ exchanges the two ends of each edge in $E_i'$.
        
    \end{enumerate}
Then $\rt(G) \le (1+cc')\cdot \rt(G')$.
\end{theorem}

\begin{proof}
Any permutation $\pi$ on graph $G'$ can be realized by a routing algorithm of at most $\rt(G')$ rounds. In one round, we swap two pebbles at vertices $u,v$ if $(u,v)\in M$ for some matching $M \subseteq E'$ of $G'$. Now we demonstrate how to implement this round on the graph $G$.
We partition edges in $M$ into at most $c+1$ sets:
\begin{align}
    & M'_0:=\{(u,v): \forall (u,v)\in M \cap E\},\\
    & M'_i:=\{(s,t): \forall (s,t)\in M \cap E_i'\}, \quad i = 1, 2, \ldots, c.
\end{align}
It takes one round to swap pebbles on $u$ and $v$ for all $(u,v)\in M_0'$ as all these $(u,v)$ are also edges in $G$. {For each $i\in [c]$, the edges in $M_i'$ form a permutation $\pi_i$, which can be realized by a routing algorithm of $c'$ rounds by the second condition.} 
Therefore, one round of routing on graph $G'$ can be implemented by a routing of $(1+cc')$ rounds on graph $G$. Repeating this for all $\rt(G')$ rounds gives $\rt(G) \le (1+cc')\cdot \rt(G')$. 
\end{proof}

One simple scenario in which the second condition in Theorem \ref{thm:reduction_routing} holds is that there are vertex-disjoint paths $P_{uv}$ connecting $u$ and $v$ in $G$ of length at most $c'$ for all $(u,v)\in E_i'$, for each $i\in [c]$. Indeed, in this case, the routing algorithm in $rt(G,\pi_i)$ can be done by simply following these paths $P_{uv}$, which are vertex-disjoint and thus enable parallel routing. 
Next, we utilize this fact to design routing algorithms for the cycle-grids and brick walls.

\begin{figure}
    \centering
    \includegraphics[width=\textwidth]{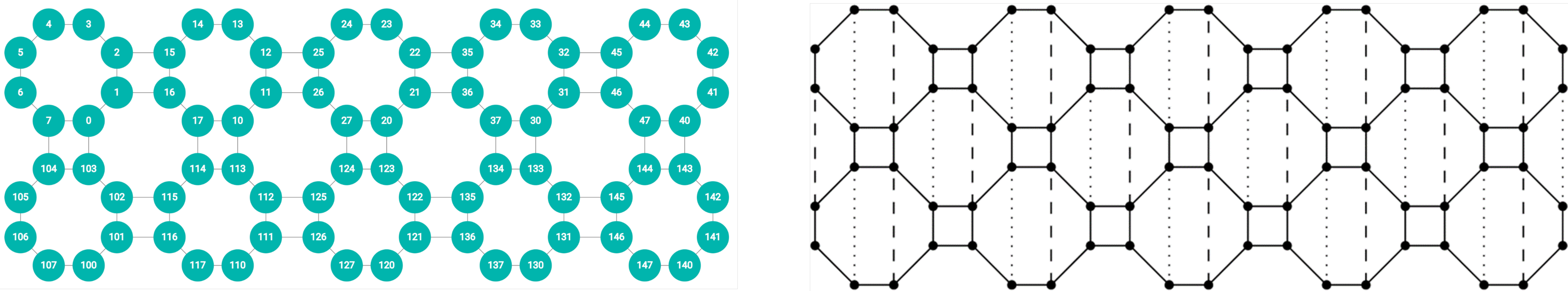}
    \caption{Left: The layout of 80 qubits in Rigetti's Aspen-M chip series. Right: A cycle-grid graph and a grid constructed from it. The cycle-grid graph consists of vertices and solid edges. A grid of size $4\times 20$ is constructed by adding dotted and dashed edges.} 
   \label{fig:rigetti}
\end{figure}

Rigetti's Aspen-M chip series has the qubit connectivity graph as in Fig. \ref{fig:rigetti}, which is a $2\times 5$ grid with each (super)node being a cycle of length 8, and adjacent (super)nodes connected by two edges. The grid sizes and the cycle length can vary. The routing and circuit compilation for such graphs can be reduced to those for grids by inserting edges. 
Specifically, we add two vertical edges to each cycle, which results in a grid of size $4\times 20$, as shown in Fig. \ref{fig:rigetti} (right). The newly added edge set is partitioned into two disjoint edge sets, the dotted edge set $E_1$ and the {dashed} edge set $E_2$. For each $i\in [2]$, all edges $(u,v) \in E_i$ have paths of length $3$ between $u$ and $v$ in the cycle-grid graph, and all paths are vertex-disjoint. Thus, we can apply Theorem \ref{thm:reduction_routing} with parameters $c=2$ and $c'=3$, obtaining routing algorithms and circuit compilation from those on the 2D grid of size $4\times 20$.

We can also apply this reduction result on brick walls. For integers $n_1,n_2\ge 1$, $b_1\ge 2$, $b_2\ge 3$, $b_1<b_2$ and $b_2$ odd, the $(n_1,n_2,b_1,b_2)$-brick wall graph contains $n_1$ layers of $n_2$ ``bricks'', with each brick being a rectangle containing $b_1$ vertices on each ``vertical'' edge and $b_2$ vertices on each ``horizontal'' edge. In IBM's brick wall chips \cite{IBMQ}, $b_1=3$ and $b_2=5$.
\begin{figure}[]
    \centering
    \includegraphics[width=0.65\textwidth]{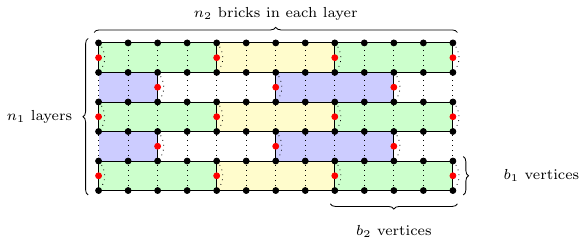}
   \caption{
   {An $(n_1,n_2,b_1,b_2)$-brick wall graph and a grid constructed from it. The $(n_1,n_2,b_1,b_2)$-brick wall graph consists of black and red vertices and solid edges.} A grid is constructed by (i) removing the vertices in the middle of each vertical edge (the red vertices), and (ii) adding dotted edges.
   Bricks are divided into $4$ groups, indicated by the green, white, yellow and blue colors, where the bricks of the same color are vertex-disjoint.
   }
   \label{fig:brickwall_grid}
\end{figure}

The reduction of brick walls to grids is more complicated because the reduction changes the vertex set. Yet a reduction in the same spirit can still be achieved as follows. 

\begin{theorem}
    For an $(n_1,n_2,b_1,b_2)$-brick wall ${\tt Brickwall}_{n_1,n_2}^{b_1,b_2}$, we have 
    \begin{equation}
        \rt({\tt Brickwall}_{n_1,n_2}^{b_1,b_2})=O((b_1+b_2)(n_1+b_2n_2)).
    \end{equation}
\end{theorem}

\begin{proof}
For ease of presentation, we color the vertices in ${\tt Brickwall}_{n_1,n_2}^{b_1,b_2}$ black and red as follows. Recall that each brick is a rectangle containing $b_1$ vertices on each vertical edge and $b_2$ vertices on each horizontal edge. Color the $2b_2$ vertices on the horizontal edges black, and the rest $2(b_1-2)$ vertices red; see Fig. \ref{fig:brickwall_grid}. 
We will first show that any permutation $\sigma$ on \textit{black} vertices has
\begin{equation}\label{eq:brackwall-black}
  \rt({\tt Brickwall}_{n_1,n_2}^{b_1,b_2},\sigma) = O((b_1+b_2)(n_1+b_2n_2)).  
\end{equation} 
We will do this by reducing the routing problem to that on a 2D-grid. To do so, we first remove the $b_1-2$ red vertices in each vertical edge, and then add edges to connect the vertically aligned two vertices in each brick---the curved and straight dotted edges in Fig. \ref{fig:brickwall_grid}. In this way, the graph ${\tt Brickwall}_{n_1,n_2}^{b_1,b_2}$ is transformed to ${\tt Grid}_{n_1+1,n_2b_2-n_2+1}$, a grid of size $(n_1+1)\times (n_2b_2-n_2+1)$. And note that all black vertices are in this grid, thus $\sigma$ is also a permutation on vertices of ${\tt Grid}_{n_1+1,n_2b_2-n_2+1}$. Now to implement $\sigma$ in ${\tt Brickwall}_{n_1,n_2}^{b_1,b_2}$, we simulate the routing protocol $\mathcal P$ in $\rt({\tt Grid}_{n_1+1,n_2b_2-n_2+1},\sigma)$. In each round of $\mathcal P$, there are swaps of pebbles on two vertices $(i,j)$ for some disjoint edges $(i,j)$ in the grid graph. If $(i,j)$ is an horizontal edge, it can be directly carried out in the brick wall graph as well, because $(i,j)$ is also an edge in the brick wall graph. Therefore it suffices to efficiently swap pebbles on $(i,j)$ for dotted edges $(i,j)$. 

For better parallelization, we partition the bricks into four disjoint groups, indicated by the four colors in Fig. \ref{fig:brickwall_grid}: Use two colors for the odd layers and the other two colors for the even layers; inside each layer, use the two colors alternately. Note that the bricks in the same color do not overlap on vertices or edges, thus the routing can be made parallel for different bricks of the same color. 
Swapping pebbles on $(i,j)$ for any number of dotted edges inside each  brick in ${\tt Brickwall}_{n_1,n_2}^{b_1,b_2}$ is a routing problem inside the brick, which is $C_{2b_2+2b_1-4}$, a cycle of length $2b_2+2b_1-4$. The routing number for this cycle is $\rt(C_{2b_2+2b_1-4})=O(b_1+b_2)$ by path result in Lemma \ref{lem:path_rout}, and note that routing for all the bricks of the same color can be done in parallel. Thus overall one round of the pebble swapping on ${\tt Grid}_{n_1+1,n_2b_2-n_2+1}$ can be implemented in $O((b_1+b_2)(n_1+b_2n_2))$ rounds on ${\tt Brickwall}_{n_1,n_2}^{b_1,b_2}$, proving Eq. \eqref{eq:brackwall-black}.
    
Now let us consider a general permutation on vertices of ${\tt Brickwall}_{n_1,n_2}^{b_1,b_2}
$. First, according to Lemma \ref{lem:permutation_decom}, any permutation can be decomposed as two transpositions, and we can implement them one by one. Now consider a transposition $\pi=(a_1,a_2)(a_3,a_4)\cdots (a_{2k-1},a_{2k})$ on ${\tt Brickwall}_{n_1,n_2}^{b_1,b_2}
$, where $a_1,a_2,\ldots, a_{2k}$ are all distinct. This transposition has three parts, according to the color of the two vertices in each pair. More precisely, we decompose $\pi=\pi_1\circ\pi_2\circ\pi_3$, where $\pi_1$, $\pi_2$ and $\pi_3$ consist of all the transpositions $(a_{2i-1},a_{2i})$ in $\pi$ where $a_{2i-1}$ and $a_{2i}$ are both black, both red, and one black and one red, respectively. We discuss the implementation of them one by one.
\begin{itemize}
    \item Permutation $\pi_1$: Since $\pi_1$ are on black vertices only, we can use the protocol as discussed above to achieve $\rt({\tt Brickwall}_{n_1,n_2}^{b_1,b_2},\pi_1) = O((b_1+b_2)(n_1+b_2n_2))$.

    \item Permutation $\pi_2$:  Consider all bricks of the green color, and perform the following in parallel. For each brick $b$, collect all vertices $a_{b1},\ldots, a_{bj}$ that are both in the left side of this brick and appear in $\pi_2$. Swap (pebbles on) all these $j$ vertices with $j$ distinct black vertices $c_{b1}, \ldots, c_{bj}$ in the upper side in the same brick---this can be conducted because of the assumption that $b_2 \ge b_1$. We also do this for the other three colors one by one, then all red vertices are moved to black ones (except for the ones on the right boundary of the green and blue bricks, which can be easily handled later). This moving takes $O(b_1+b_2)$ rounds. The problem reduces to the first case $\pi_1$, which can be solved by $O((b_1+b_2)(n_1+b_2n_2))$ rounds. Note that in this process, the pebbles originally on $c_{b1}, \ldots, c_{bj}$ remain still in $a_{b1},\ldots, a_{bj}$ (because these pebbles are now on red vertices but $\pi_1$ only swaps black vertices). We then swap pebbles on $a_{b1},\ldots, a_{bj}$ and $c_{b1}, \ldots, c_{bj}$ back. The overall cost is $O((b_1+b_2)(n_1+b_2n_2))$ rounds. (The red ones on the right boundary of the green and blue bricks can then be handled similarly with at most the same cost.)

    \item Permutation $\pi_3$: For transpositions $(a_{2i-1},a_{2i})$, assume without loss of generality that $a_{2i-1}$ is black and $a_{2i}$ is red. Let $p_{2i-1}$ and $p_{2i}$ denote the pebbles at vertices $a_{2i-1}$ and $a_{2i}$ respectively. 
        First, we permute each pebble  $p_{2i-1}$ from $a_{2i-1}$ to a black vertex $b_{2i-1}$ within the same brick of red vertices $a_{2i}$. (If $a_{2i-1}$ and $a_{2i}$ are already in the same brick then we do not need to move it). This is possible because $b_2 \ge b_1$.   This permutation among black vertices only can be implemented in $O((b_1+b_2)(n_1+b_2n_2))$ rounds.
        Second, we exchange pebbles $p_{2i}$ at red vertices $a_{2i}$ and pebbles $p_{2i-1}$ at black vertices in the same brick. This can be implemented in $O(b_1+b_2)$ rounds. 
        Third, we apply the inverse permutation of the first step in $O((b_1+b_2)(n_1+b_2n_2))$ rounds.
        The total rounds of $\pi_3$ is $2\cdot O((b_1+b_2)(n_1+b_2n_2))+O(b_1+b_2)=O((b_1+b_2)(n_1+b_2n_2))$. 
    \end{itemize}
    Putting all steps together, the permutation $\pi=\pi_1\circ \pi_2 \circ \pi_3$ on brick wall ${\tt Brickwall}^{b_1,b_2}_{n_1,n_2}$ can be implemented in $O((b_1+b_2)(n_1+b_2n_2))$ rounds.
\end{proof}

\section{Discussion}
\label{sec:discussion}
We have fully characterized the depth overhead when compiling a quantum circuit with all-to-all qubit connections to one under a connectivity graph constraint for any graph, by a well-studied graph measure called routing number. We also developed compiling algorithms to achieve the asymptotically optimal depth overhead. This enables us to utilize existing routing algorithms for various specific graphs such as paths, grids and trees, to construct hardware-compliant compilers. Constraint graphs like IBM's brick walls and Rigetti's cycle-grids can also be handled via some simple reductions. The characterization of the depth overhead can also be utilized in the design of quantum processors when a small depth overhead is crucially needed.

\bibliographystyle{quantum}
\bibliography{depthoverhead}

\newpage
\appendix
\section{Hardness of computing depth overhead}
\label{app:hardness}
In this section, we show the hardness of computing the depth overhead of a given circuit $C$ and a constraint graph $G$. 

First, we present a few initial concepts that will be utilized in the subsequent proof.
\paragraph{Quantum gates and circuits.} A SWAP gate implements the unitary which swaps the basis $|01\rangle$ and $|10\rangle$ and keeps $|00\rangle$ and $|11\rangle$ unchanged. A SWAP gate can be implemented by three CNOT gates. 

Throughout this paper, we consider quantum circuits on $n$ qubits. Qubit connectivity constraint is specified by an undirected graph $G = (V,E)$, where $V$ is the set of $n$ vertices, often identified with the $n$ qubits, and $E$ is the set of edges, specifying pairs of qubits that two-qubit gates can act on. A circuit without any qubit connectivity constraint is called an \textit{unconstrained circuit}. A circuit made of SWAP gates only is called a \textit{SWAP circuit}. A SWAP circuit $C$ implements a permutation on qubits in the sense that $C\ket{x_1 x_2 \ldots x_n} = \ket{x_{\pi(1)} x_{\pi(2)} \ldots x_{\pi(n)}}$ for some qubit permutation $\pi\in S_n$. 

For any qubit connectivity graph $G$ and any permutation $\pi$ of vertices, there is a natural one-one correspondence between the following two sets:
\begin{multline}\label{eq:circuit-routing-correspondence}
    \{\text{$G$-compliant SWAP circuits $C$ of depth $d$ realizing $\pi$}\} \\ 
    \leftrightarrow \{\text{routing algorithm $A$ of $d$ rounds on $G$ realizing $\pi$}\}. 
\end{multline}
Indeed, each layer of a SWAP circuit $C$ consists of SWAP gates on distinct qubits, which 1-1 correspond to one round of a routing algorithm consisting of swapping pebbles. Note that $C$ and $A$ implement the same permutation, and are compliant with the same graph $G$. 

Two quantum circuits $C$ and $C'$ are equivalent if they realize the same unitary operation. In this paper, we are interested in a special \textit{SWAP equivalence}, denoted by $C' \sim C$, in which $C'$ is obtained from $C$ by inserting SWAP gates and applying $C$'s gates on permuted qubits. More specifically, suppose that $C$ = $\prod_{i=1}^d L_i$ where $L_i$ is the $i$-th layer, made of one- and two-qubit gates on distinct qubits. A circuit $C'$ satisfies $C'\sim C$ if $C' = P(\pi)^\dagger \prod_{i=1}^d L_i'$, where $L_i'$ is a circuit similar to $L_i$, but possibly with SWAP circuits inserted between the gates in $L_i$: If a gate $U$ in $L_i$ is on qubits $p$ and $q$, then in $L_i'$ it should be $U$ applied to qubits $\sigma(p)$ and $\sigma(q)$, where $\sigma$ is the permutation on qubits realized by the inserted SWAP gates in circuit $C'$ from beginning to right before $U$;
$\pi$ is the permutation of qubits by \textit{all} SWAP gates in circuit $\prod_{i=1}^d L_i'$, and $P_{\pi}^\dagger: \ket{x_{\pi(1)} x_{\pi(2)} \ldots x_{\pi(n)}} \mapsto \ket{x_1 x_2 \ldots x_n}$ is just to permute the qubits back to their original positions.

\paragraph{Languages and decision problems.} 
A \textit{language} $L$ is {a set of strings of finite lengths over a finite alphabet $\Sigma$.} The corresponding \textit{decision problem $D_L$} is to decide whether an input instance $x$ belongs to the language $L$. A language $L$ (or its decision problem $D_L$) is in NP {if there is a polynomial-time Turing machine $M$ such that the following two statements are equivalent for any input $x$: (1) $x\in L$, and (2) there is a certificate string $y$ of length polynomial in that of $x$ such that $M$ accepts $(x,y)$ \cite{arora2009computational}.}  A language $L$ or its decision problem $D_L$ is \textit{NP-complete} if $L$ is in NP and any other NP problem can be reduced to $L$ in polynomial time. If an NP-complete language $L$ reduces to an NP language $L'$ {in polynomial time}, then $L'$ is also NP-complete.

It is known that computing the routing number $\rt(G,\pi)$ of a given connected graph $G$ and a permutation $\pi$ is hard, as the following lemma states. 
\begin{lemma}[Routing deciding problem, \cite{banerjee2017new}]\label{lem:routing_npc}
  Deciding whether $\rt(G,\pi)\le k$ on a given graph $G$, vertex permutation $\pi$, and any integer $k\ge 3$ is NP-complete.  
\end{lemma}

Second, using this result, we can show the hardness of computing the depth overhead.

\begin{theorem}[Depth overhead deciding problem]\label{thm:depthoverhead_npc}
    Deciding whether $\doh(G,C)\le \alpha$ on a given $n$-node graph $G$, a {quantum circuit $C$}, and a rational number $\alpha$ (with description length polynomial in $n$) is NP-complete. 
\end{theorem}
{
\begin{proof}
    We first show that the hardness of deciding whether $\doh(G,C)\le \alpha$ is no more than that of NP, or more precisely, the language $\{(G,C,\alpha): \doh(G,C)\le \alpha\}$ belongs to NP. Given a Yes input instance, i.e. a tuple $(G,C,\alpha)$ where $\doh(G,C)\le \alpha$, we can use as a certificate a minimum-depth $G$-compliant circuit $C'$ which is equivalent to $C$ and obtained from $C$ by inserting SWAP gates only. By definition of $\doh(G,C)$, the depth of $C'$ is $d(C') = \doh(G,C)\cdot d(C)$, where $d(C)$ is the depth of $C$. The verification algorithm needs to check that (1) $d(C')/d(C) \le \alpha$, and (2) $C'\sim C$. The first is trivially done in polynomial time, and the second is also easy as one can keep track of all permutations and check whether for each gate $U$ in $C$ and the corresponding gate $U$ in $C'$, whether the latter applies on the $\pi$-permuted qubits with $\pi$ being the qubit permutation from the beginning of $C'$ up to gate $U$.

    Next, for any given graph $G=(V,E)$, permutation $\pi$ and integer $k\ge 3$ of the routing deciding problem $\rt(G,\pi)\le k$, we will efficiently construct a circuit $C_\pi$ under no graph constraint, and define a rational number $\alpha$, such that $\rt(G,\pi)\le k$ if and only if $\doh(G,C_\pi)\le \alpha$. 

    Circuit $C_\pi$ and rational number $\alpha$ are determined as follows.
    \begin{itemize}
        \item For any permutation $\pi$ of routing deciding problem $\rt(G,\pi)\le k$, compute $\pi^2$.

        \item If $\pi^2=id$, then $\pi$ is a transposition and can be written as \[\pi=(a_1,a_2)(a_3,a_4)\cdots (a_{2m-1},a_{2m})\] for some $m\le n/2$ and distinct $a_1,a_2,\cdots,a_{2m}\in V$. 
        
        \begin{itemize}
            \item Define an unconstricted circuit $C_\pi = \prod_{i=1}^m \text{SWAP}(a_{2i-1},a_{2i})$, which consists of $m$ SWAP gates. Since all $a_i$'s are distinct, the depth of $C_\pi$ is clearly $d(C_{\pi})=1$.

            \item Let $\alpha = k$.
        \end{itemize}

        \item If $\pi^2\neq id$, then $\pi$ is not a transposition. According to Lemma \ref{lem:permutation_decom}, $\pi$ can be decomposed as two transpositions $\pi = \sigma_1\circ \sigma_2$. As discussed above, we can construct two SWAP circuits of depth 1 for $\sigma_1$ and $\sigma_2$. 
        \begin{itemize}
            \item Combining these SWAP circuits gives an unconstrained SWAP circuit $C_\pi$ with depth $d(C_\pi)=2$.

            \item Let $\alpha = k/2$.
        \end{itemize}
    \end{itemize}
    Observe that in either case, the circuit $C_\pi$ realizes permutation $\pi$ in the natural sense that $C_\pi\ket{x_1 x_2 \ldots x_n} = \ket{x_{\pi(1)} x_{\pi(2)} \ldots x_{\pi(n)}}$, and that $\alpha \cdot d(C_\pi) = k$. Furthermore, the above algorithm mapping $(G,\pi,k)$ to  $(G, C_\pi,\alpha)$ can clearly be implemented in polynomial time. Now we will show that 
    \[\rt(G,\pi)\le k \text{ if and only if }\doh(G,C_\pi)\le \alpha.\] 
    First, we claim that the following equality holds for this particular $C_\pi$.
    \begin{align}\label{eq:npc-reduction}
        \min \big\{d(C_\pi'): C_\pi' \sim C_\pi, \text{ and $ C_\pi'$ is $G$-compliant}\big\} = \rt(G,\pi).
    \end{align}
    \begin{itemize}
        \item ``$\ge$'': Take any $C_\pi'$ that achieves the minimum of the left-hand side. Since $C_\pi$ realizes $\pi$ and $C_\pi' \sim C_\pi$, we have that $C_\pi'$ realizes $\pi$ as well. Also note that $C_\pi$ is a SWAP circuit, and $C_\pi'$ is obtained from $C_\pi$ by inserting SWAP gates, therefore $C_\pi'$ is also a SWAP circuit. By the correspondence in Eq. \eqref{eq:circuit-routing-correspondence}, $C_\pi'$ induces a $d(C_\pi')$-round routing algorithm realizing $\pi$, thus $\rt(G,\pi)$ as the minimum number of rounds of such protocols is at most $d(C_\pi')$. 
        \item ``$\le$'': Take any routing algorithm $A$ with $\rt(G,\pi)$ rounds. Again by the correspondence in Eq. \eqref{eq:circuit-routing-correspondence}, $A$ induces a $\rt(G,\pi)$-depth SWAP circuit $C(A)$ realizing $\pi$. Any $C_\pi'\sim C_\pi$ only inserts SWAP gates to $C_\pi$, thus is also a SWAP circuit. Therefore, the set in the left-hand side contains all $G$-compliant SWAP circuits realizing $\pi$, including $C(A)$. Thus the minimum depth $d(C_\pi')$ is at most the depth of $C(A)$, which is $\rt(G,\pi)$.
    \end{itemize}
    Now with Eq. \eqref{eq:npc-reduction}, we can see that 
    \begin{align*}
        & \ \doh(G,C) \le \alpha \\
        \Leftrightarrow & \ \min \big\{d(C_\pi'): C_\pi' \sim C_\pi, \text{ and $ C_\pi'$ is $G$-compliant}\big\}/d(C_\pi) \le \alpha & (\text{by definition of } \doh(G,C)) \\
        \Leftrightarrow & \ \frac{\rt(G,\pi)}{d(C_\pi)} \le \alpha & (\text{by Eq.~\eqref{eq:npc-reduction}})\\
        \Leftrightarrow & \ \rt(G,\pi) \le k & (\text{by definition of } \alpha)
    \end{align*}
    This completes the proof.
\end{proof}

\end{document}